\documentclass[a4paper]{cas-dc}

\usepackage[numbers]{natbib}

\def\tsc#1{\csdef{#1}{\textsc{\lowercase{#1}}\xspace}}
\tsc{WGM}
\tsc{QE}

\newproof{pf}{\textbf{Proof}}

\newtheorem{rem}{Remark}
\newtheorem{assum}{Assumption}
\newtheorem{thm}{Theorem}
\newtheorem{lem}{Lemma}
\newtheorem{prop}{Proposition}

\usepackage{subcaption}
 
\usepackage{graphicx}
\graphicspath{{figures/}}

\begin{document}
\let\WriteBookmarks\relax
\def\floatpagepagefraction{1}
\def\textpagefraction{.001}

\shorttitle{}    

\shortauthors{}  

\title [mode = title]{Stabilization of First-Order Partial Integro-Differential Equations with Concurrent Input and State Delays}  

\tnotemark[1] 

\tnotetext[1]{This work was supported by the National Natural Science Foundation of China (62173084), the Project of Science and Technology Commission of Shanghai Municipality, China (23ZR1401800), the Fundamental Research Funds for the Central Universities (CUSF-DH-T-2023026).}

\author[1]{Sanguan Zhong}

\ead{zsgguan@mail.dhu.edu.cn}

\affiliation[1]{organization={School of Information and Intelligent Science, Donghua University},
            city={Shanghai},
            postcode={201620}, 
            country={China}}

\author[1]{Jie Qi}

\cormark[1]

\ead{jieqi@dhu.edu.cn}

\cortext[1]{Corresponding author}

\begin{abstract}
This paper considers boundary stabilization problems for a first-order hyperbolic partial integro-differential equation (PIDE) subject to concurrent input and state delays. The coexistence of these two types of delays complicates control design, especially under the case of large input delay that requires to predict more state information. A backstepping-based boundary controller is developed to achieve stabilization and delay compensation. The design relies on two affine Volterra transformations involving both Fredholm- and Volterra-type integral terms, which results in a four-PIDE kernel equations. To establish their well-posedness, the kernel domain is partitioned along characteristic lines into a finite number of triangular subregions, and the kernel solution is constructed successively in the preceding subregion to the next one. The finite-time stability of the resulting closed-loop system is established. Numerical simulations are provided to demonstrate the effectiveness of the proposed controller. 
\end{abstract}


\maketitle

\section{Introduction}
 
    First-order hyperbolic partial integro-differential equations (PIDEs) frequently arise in engineering applications, including traffic flows \cite{zhang2022boundary}, heat exchangers \cite{xu2002exponential,ghousein2020adaptive}, pipe flows \cite{ghidaoui2005review,housiadas2024pressure}, and chemical reaction kinetics in plug flow reactors \cite{steinboeck2016real,karafyllis2002control,xu2018optimal,bastin2016stability}.  In such processes, time delays commonly emerge due to the transport of mass, energy, or information. These delays may appear not only in control inputs but also in system states and measurements. For instance, in plug flow reactors with recycle, a portion of the reactor effluent is recirculated as a feedback stream, introducing state delays \cite{qi2021output,oh1974study}. However, time delays are well known to degrade system stability \cite{guo2008boundary}, thereby necessitating appropriate compensation strategies. While progress has been made in control of ordinary differential equations (ODEs) with both input and state delays (see \cite{kharitonov2017prediction}), analogous results for partial differential equations (PDEs) remain unexplored.
    
    A wide range of studies have investigated the control of hyperbolic partial differential equations (PDEs), sometimes in conjunction with ODEs. The pioneering application of the backstepping method to first-order linear PIDE is presented in \cite{krstic2008backstepping}, followed by its extension to $2\times 2$ first-order linear hyperbolic systems in \cite{vazquez2011backstepping}. The quasilinear case is further explored in \cite{vazquez2011local,coron2013local}. Subsequent works address more general settings, including $n+1$ linear heterodirectional hyperbolic systems \cite{di2013stabilization}, homodirectional and general heterodirectional $n+m$ coupled linear systems \cite{hu2016control}, and even ensembles of $\infty+1$ coupled linear hyperbolic PDEs \cite{alleaume2025ensembles,humaloja2025stabilization}. In addition, related topics such as adaptive control and observer design \cite{bernard2014adaptive,anfinsen2019adaptive}, disturbance rejection \cite{aamo2013disturbance}, and output regulation \cite{deutscher2017backstepping,xu2023adaptive} have also been addressed. For more research on hyperbolic systems, the reader is referred to the recent survey \cite{vazquez2026backstepping} and the references therein.
    
    For time delayed PDE systems, various approaches have been proposed, including Lyapunov-Krasovskii functional method \cite{selivanov2018delayed}, spectral method \cite{lhachemi2021feedback,lhachemi2022predictor,lhachemi2023boundary}, and observer-based approach\cite{katz2021delayed}. Among these, the backstepping technique turns out to be effective to compensate for arbitrarily long delays, such as for constant input delay \cite{krstic2009control}, spatially varying input delay \cite{qi2021compensation}, state delay \cite{hashimoto2016stabilization}, unknown input delay \cite{wang2022delay}, and input and output delays \cite{deutscher2020fredholm}. For the first-order hyperbolic PDE, \cite{sano2019boundary} solves the stabilization problem with input delay, while a delay-adaptive boundary controller is developed for a $2\times2$ coupled hyperbolic PDE-ODE system with unknown input delay in \cite{wangji2024delay}. Other relevant results include \cite{auriol2018delay,qi2023delay,zhang2021compensation,qi2021output}. 

Most of the aforementioned contributions are restricted to systems involving either input delays or state delays alone. The simultaneous compensation of input and state delays remains a challenging problem in delay-system control, as discussed in \cite{krstic2009delay,richard2003time}. Although several related problems can be transformed into comparable forms \cite{redaud2022stabilizing}, they do not directly address the PIDE setting considered here. For example, \cite{jean2022explicit} develops an explicit prediction-based controller for linear finite-dimensional difference equations with a distributed state delay and an input delay. Its applicability, however, relies on a delay-ordering condition requiring the input delay to be no less than the horizon of the distributed delay. Their approach, remaining confined to the ODE framework, does not--at least in a straightforward manner--extend to the reverse delay configuration. 

In contrast, this paper addresses the concurrent input-and-state-delay problem directly within a PDE framework. By considering a first-order hyperbolic PIDE subject to both types of delays without considering the delay-ordering, the proposed design offers a possible approach to the stabilization of PDE systems with simultaneous input and state delays.

The main difficulty arises from the coupled effects of the two delays \cite{krstic2009delay}. In particular, the state-feedback law required to compensate for the input delay depends on delayed state information, while the state delay itself modifies the prediction structure needed for input-delay compensation. As a result, the two delays cannot be treated independently, and the standard backstepping transformation is no longer directly applicable. To address this issue, we introduce a pair of integral transformations, each involving both Fredholm- and Volterra-type integral terms.
This construction is different from the Fredholm-transformation approach in \cite{redaud2022stabilizing} where two preliminary Volterra transformations are applied, and the Fredholm transformation is then introduced as a second step.

The resulting kernel equations are coupled and considerably complex. Their well-posedness analysis is particularly challenging in the presence of a large input. In this case, conventional arguments based solely on the method of characteristics and successive approximations are no longer sufficient, because the enlarged integration domain prevents the direct convergence of the iterative series. To overcome this difficulty, we introduce a domain-partitioning strategy: the kernel domain is divided along characteristic lines into several triangular subregions, and the solution on each subregion is constructed successively by using the solution obtained on the preceding subregion as boundary data for the next one.

Furthermore, for the pair of coupled kernel equations appearing in the transformation associated with the delayed input, each governed by two boundary conditions, we divide the well-posedness analysis into two cases, depending on whether the input delay is smaller or larger than the state delay. The characteristics method and the infinite induction energy series method  \cite{chen2024backstepping}, are then employed, respectively, to treat these two cases.

Finally, we establish both exponential and finite-time stability of the closed-loop system under the proposed delay-compensating boundary controller, and a numerical simulation is presented to demonstrate the effectiveness of the developed control strategy.

    The main contributions of this paper are as follows:
    \begin{enumerate}
    \item 
    As discussed in \cite{krstic2009delay,richard2003time}, one single important open area in the control of delay systems is the control of systems with simultaneous input and state delays. To the best of the authors’ knowledge, this paper provides the first direct result on the stabilization of unstable PDE systems subject to both input and state delays.
    \item
    The kernel equations of the system with coexisting input and state delays exhibit enlarged integration domains as the input delay increases, under which the method of successive approximations fails to produce a convergent series. To overcome this difficulty, we propose a domain-partitioning strategy that divides the enlarged integration domain into multiple finite triangular subregions. In contrast to \cite{redaud2022stabilizing}, this approach yields explicit solutions expressed as convergent series.
    \end{enumerate}
    
    The rest of this paper is organized as follows. The problem formulation is shown in Section \ref{problem statement}. The backstepping boundary control design is presented in Section \ref{control design}, where we state the main result. The well-posedness of the kernel equations, which is the main technical difficulty of this paper, is proved in Section \ref{well-posedness analysis}. We perform closed-loop stability analysis in Section \ref{stability analysis}. The effectiveness of the proposed control is illustrated with a numerical example in Section \ref{numerical simulation}. Finally, Section \ref{conclusion} concludes this paper.
    
    \emph{Notation}: Define norm for $f\in L^2(0,1)$ and norm for $g\in L^\infty(0,1)$ as follows:
    \begin{align*}
        \|f\|&:=\left(\int_0^1 f^2(x)dx\right)^{\frac{1}{2}}, & \|g\|_\infty&:= \operatorname*{ess\,sup}_{x\in[0,1]} |g(x)|. 
    \end{align*}
    Define domains: $\mathcal{D}_1=\{(x,y)\in[0,1]^2: 0\le x\le y\le1\}$, $\mathcal{D}_2=\{(x,y)\in[0,1]^2\}$.

\section{Problem statement}\label{problem statement}
    Consider the first-order hyperbolic PIDE which is subject to boundary input delay $D_1>0$ and state delay $D_2>0$, shown in Fig. \ref{fig:system diagram},
    \begin{align}
    u_t(x,t)=&-u_x(x,t)+c(x)u(1,t-D_2)  \notag \\
    &\quad+\int_x^1f(x,y)u(y,t)dy, \label{eq: u} \\
    u(0,t)=&U(t-D_1). \label{eq: u boundary}
    \end{align}
    for $(x,t)\in (0,1)\times \mathbb{R}^+$ with the initial condition $u(x,0)=u_0\in L^2(0,1)$. The historical input is $U(\theta)={\vartheta}_1\in L^2(-D_1,0)$, and the historical state is $u(1,\theta)={\vartheta}_2\in L^2(-D_2,0)$. For large positive $c(x)$ and $f(x,y)$, the system is open-loop unstable \cite{krstic2008backstepping}. The objective is to design a boundary controller $U(t)$ that stabilizes the system while simultaneously compensating for both input and state delays.  
    \begin{figure}[htpb]
        \begin{center}
            \includegraphics[width=\columnwidth]{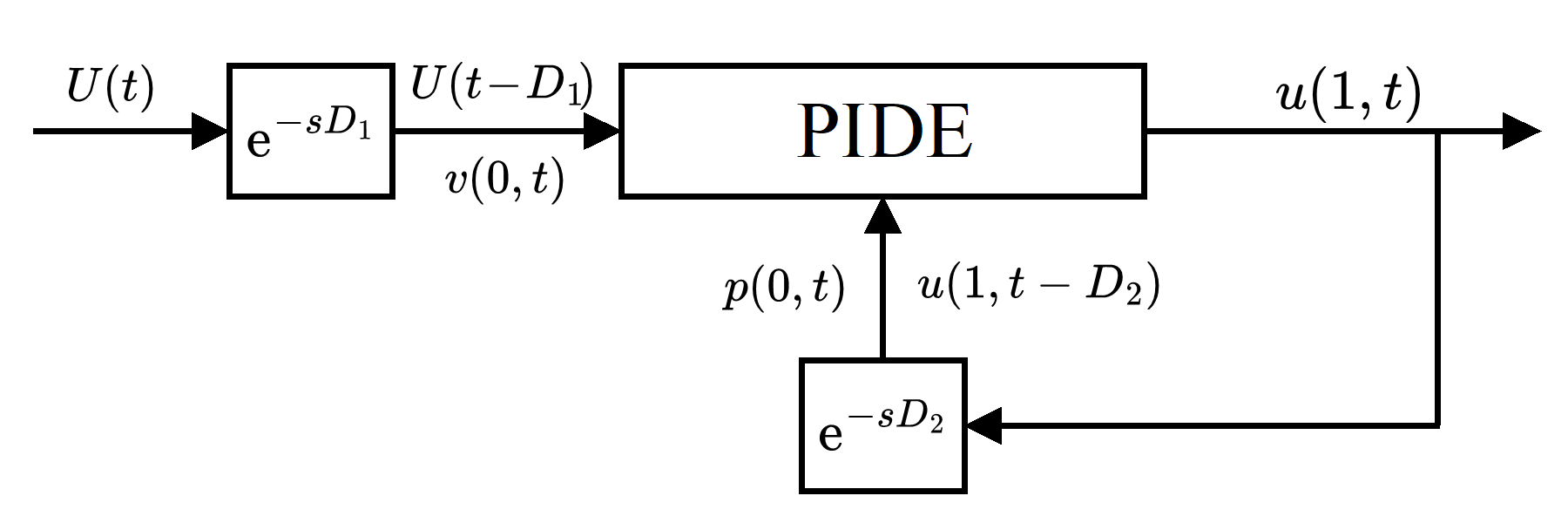}            
        \end{center}
        \caption{Schematic diagram of the system with delay dynamics.}
        \label{fig:system diagram}
    \end{figure}
    
    \begin{rem}
       This work is motivated by plug flow reactor systems, which provide a representative application prototype for the first-order hyperbolic PIDE considered in this paper. The spatial evolution of the reactant temperature is naturally governed by an advection-dominated PDE \cite{oh1974study,qi2021compensation,zhang2021compensation,sano2019boundary}. In practical reactor operation, the recycle of a part of the outlet stream introduces a state delay through finite-speed transport in the recycle line, while feed-line transportation lags and thermal inertia give rise to an input delay at the inlet boundary. 
    \end{rem}
    
    \begin{assum}\label{assumption}
        Assume that the coefficient functions $c\in C([0,1])$, $f\in C(\mathcal{D}_1)$, and $c$ satisfies $c(0)=c(1)=0$.
    \end{assum}
    \begin{rem}
       Setting the coefficient $c(x)$ to zero at both ends ensures the continuity of the kernel functions within the domain. For systems with input delay, this requires $c(0)=0$, consistent with \cite{sano2019boundary}, while for systems with state delay,  $c(1)=0$ is imposed, as in \cite{qi2021output}. However, this constraint on $c(x)$  can be relaxed to allow piecewise continuous kernel functions, as discussed in \cite{hu2016control,hu2019boundary}. A more detailed discussion is provided in Section~\ref{well-posedness analysis}.
    \end{rem}
    
    We introduce two transport equations with distinct transport speeds in the representation of the delayed input and the delayed state, respectively, leading to the following cascaded system:
    \begin{subequations}\label{extended system}
    \begin{align}
        u_{t}(x,t)=&\!-\!u_{x}(x,t)\!+\!c(x)p(0,t)\!+\!\int_{x}^{1}\!\!f(x,y)u(y,t)dy, \label{cas u}\\
        u(0,t)=&v(0,t),\label{cas u boundary}  \\
        D_1v_{t}(s,t)=&v_{s}(s,t),  \qquad s\in[0, 1), \label{cas v}\\
        v(1,t)=&U(t),\label{cas v boundary}\\
        D_2p_{t}(s,t)=&p_{s}(s,t),   \qquad s\in[0, 1), \label{cas p}\\
        p(1,t)=&u(1,t),\label{cas p boundary}
    \end{align}
    \end{subequations}
    with initial conditions $v(s,0)=v_0(s)$, $p(s,0)=p_0(s)$, which are related with the historical values $v_0(s)=\vartheta_1(D_1(s-1)) \in L^2((0,1);\mathbb{R})$, $p_0(s)=\vartheta_2(D_2(s-1))\in L^2((0,1);\mathbb{R})$.

\section{Backstepping control design}\label{control design}
\subsection{Backstepping transformation}
    Introduce the following backstepping transformation:
    \begin{align}
        w(x,t) =& u(x,t)-\int_{x}^{1} K(x,y)u(y,t)dy \notag \\
        &-\int_{0}^{1}F(x,y)p(y,t)dy, \label{u trans} \\
        z(s,t) =& v(s,t)-\int_0^sD_1J(s-y,0)v(y,t)dy \notag \\
        &\!-\!\!\int_{0}^{1}\! J(s,y)u(y,t)dy \!-\!\! \int_{0}^{1}\! G(s,y)p(y,t)dy, \label{v trans}
    \end{align}
    where the kernel function $K$ is defined on $\mathcal{D}_1$, and the kernel functions $F$,  $J$ and $G$ are defined on $\mathcal{D}_2$. 
    Note that the transformation involves both the Volterra and Fredholm integrals, reflecting the nature of the transformation associated with delayed systems. Since the state $p$ remains unchanged, the transformation $(p,u,v)\mapsto(p,w,z)$ is reformulated as a lower-triangular operator representation:
    \begin{align}\label{direct trans}
        \begin{bmatrix}
            p \\ w \\ z
        \end{bmatrix}\!\!=\!&
        \begin{bmatrix}
            I & {\mathbf 0} & \mathbf{ 0} \\
            \mathscr{F}_1 & \mathscr{L}_1+I & {\mathbf 0} \\
            \mathscr{F}_2 & \mathscr{F}_3 & \mathscr{L}_2+I
        \end{bmatrix}\!\!
        \begin{bmatrix}
            p \\ u \\ v
        \end{bmatrix}
    \end{align}
    where $I$ is the identity operator and 
    \begin{align*}
&\mathscr{F}_1[p] := -\int_{0}^{1} F(x, y) p(y,t)  dy, \\
&\mathscr{F}_2[p] := -\int_{0}^{1} G(s, y) p(y,t)  dy, \\
&\mathscr{F}_3[u] := -\int_{0}^{1} J(s, y) u(y,t) dy, \\
&\mathscr{L}_1[u] := -\int_{x}^{1} K(x, y) u(y,t)  dy, \\
&\mathscr{L}_2[v] := -\int_{0}^{s} D_1J(s-y,0) v(y,t)  dy.
    \end{align*}
    \begin{rem}
       The Volterra operators $\mathscr{L}_1$, $\mathscr{L}_2$ are lower-triangular, making the overall transformation $(p,u,v)\mapsto(p,w,z)$ lower-triangular and hence invertible \cite{krstic2009delay}.    
    \end{rem}
    
    The first transformation \eqref{u trans} is applied to eliminate the unstable integral term and the coupling term associated with state delay in \eqref{cas u}, while \eqref{v trans} is designed to compensate for the input delay, resulting in the boundary controller $U$. Using the transformation, the system \eqref{extended system} is converted to the following target system:
    \begin{subequations}\label{target system}
        \begin{align}
             w_{t}(x,t) =& -w_{x}(x,t), \label{tar w}\\
             w(0,t) =& z(0,t), \label{tar w boundary}\\
             D_1z_t(s,t) =& z_s(s,t), \label{tar z}\\
             z(1,t) =& 0, \label{tar z boundary}\\
             D_{2}p_{t}(s,t) =& p_{s}(s,t), \label{tar p}\\
             p(1,t) =& w(1,t).\label{tar p boundary}
        \end{align}
    \end{subequations}
    The following kernel equations are satisfied by the kernels $K$ and $F$,
    \begin{align}
        K_x(x,y) =&-K_{y}(x,y)+f(x,y)-\int_{x}^{y}K(x,\theta)f(\theta,y)d\theta, \label{K}\\
        K(x,1)  =& \frac{1}{D_{2}}F(x,1),\label{K(x,1)}\\
        F_x(x,y) =& \frac{1}{D_2}F_y(x,y), \label{F}\\        
        F(x,0) =& D_2\int_x^1K(x,y)c(y)dy - D_2c(x), \label{F(x,0)} \\
        F(1,y) =& 0. \label{F(1,y)}
    \end{align}
    and the kernels $J$ and $G$ satisfy
    \begin{align}
        J_{s}(s,y) =& D_{1}J_{y}(s,y)+D_{1}\int_{0}^{y}J(s,\theta)f(\theta,y)d\theta, \label{J}\\
        J(s,1) =& \frac{1}{D_{2}}G(s,1), \label{J(s,1)}\\
        J(0,y) =& K(0,y), \label{J(0,y)}\\
        G_s(s,y) =& -\frac{D_1}{D_2}G_y(s,y), \label{G}\\
        G(s,0) =& D_2\int_0^1 J(s,y)c(y)dy, \label{G(s,0)}\\
        G(0,y) =& F(0,y).\label{G(0,y)}
    \end{align}
    The well-posedness result of the kernel equations is provided in Section \ref{well-posedness analysis}. Note that the kernel functions associated with the delayed variables satisfy pure transport equations that are boundary-coupled with other kernel functions.

    \subsection{Main result}
    Incorporating the boundary conditions \eqref{cas v boundary}, \eqref{tar z boundary} and the transformation \eqref{v trans}, and recalling that $v(s,t)=U(t-D_1(1-s))$, $p(s,t)=u(1,t-D_2(1-s))$ after $t>\max\{D_1,D_2\}$, we obtain,
    \begin{align}
        U(t)=&\int_{0}^{1}J(1,y)u(y,t)dy\notag \\
        &+\int_{t-D_2}^{t}\frac{1}{D_2}G \left(1,1-\frac{t-\tau}{D_2}\right)u(1,\tau)d\tau \notag \\
        &+\int_{t-D_1}^{t}J\left(\frac{t-\tau}{D_1},0\right)U(\tau)d\tau. \label{controller}
    \end{align}
    
    The main result of this paper is stated in the following Theorem.
    \begin{thm}\label{main result}
        Consider the system \eqref{extended system} under the controller \eqref{controller} and assume  that the initial conditions $u_0$, $v_0$, $p_0\in L^2(0,1)$ are compatible with the boundary conditions. Then the resulting closed-loop system is exponentially stable in the $L^2$-sense around its equilibrium (the origin), i.e., there exist positive constants $C>0$, $\sigma>0$ such that 
        \begin{align}
            V(t)\le C\mathrm{e}^{-\sigma t}V(0),
        \end{align}
        where 
        \begin{align}
            V(t)=\|u(\cdot,t)\|^2+\|v(\cdot,t)\|^2+\|p(\cdot,t)\|^2. \label{V}
        \end{align}
        Furthermore, the system converges to the origin in finite time 
        \begin{align}
            t_\mathrm{F}=1+D_1+D_2.  \label{tF}
        \end{align}
    \end{thm}
    The proof is presented in Section \ref{stability analysis}.

\section{Well-posedness of the kernel equations}\label{well-posedness analysis}
    The well-posedness of \eqref{K}-\eqref{F(1,y)} has been established in \cite[Prop. 1]{qi2021output}, so we only consider the well-posedness of \eqref{J}-\eqref{G(0,y)}. We first consider the equation of $J$ defined in \eqref{J}-\eqref{J(0,y)} and rewrite it as the following general case:
    \begin{align}
                M_s(s,y) =& \lambda M_y(s,y)+\int_0^yM(s,\theta)N(\theta,y)d\theta, \label{M}\\
                M(s,1)=&h(s), \label{M(s,1)}\\
                M(0,y)=&g(y), \label{M(0,y)}
            \end{align}
    where $(s,y)\in \mathcal{D}_2$, the constant $\lambda>0$, the functions $g,~h\in C([0,1])$ and $N\in C(\mathcal{D}_1)$.
    \begin{lem}\label{well-posedness M}
      Assume that the boundary conditions \eqref{M(s,1)} and \eqref{M(0,y)} are compatible, i.e., $g$ and $h$ satisfy $g(1)=h(0)$. Then the kernel equation \eqref{M}-\eqref{M(0,y)} admits a unique continuous solution $M\in C(\mathcal{D}_2)$. 
    \end{lem}   
    
        \begin{pf}
             The proof is divided into two cases: (i) $\lambda\le 1$  and (ii) $\lambda>1$. 
             \begin{figure*}[htpb]
                \begin{center}
                    \begin{subfigure}[b]{0.47\linewidth}
                        \centering
                        \includegraphics[width=\linewidth,height=200pt,keepaspectratio=true]{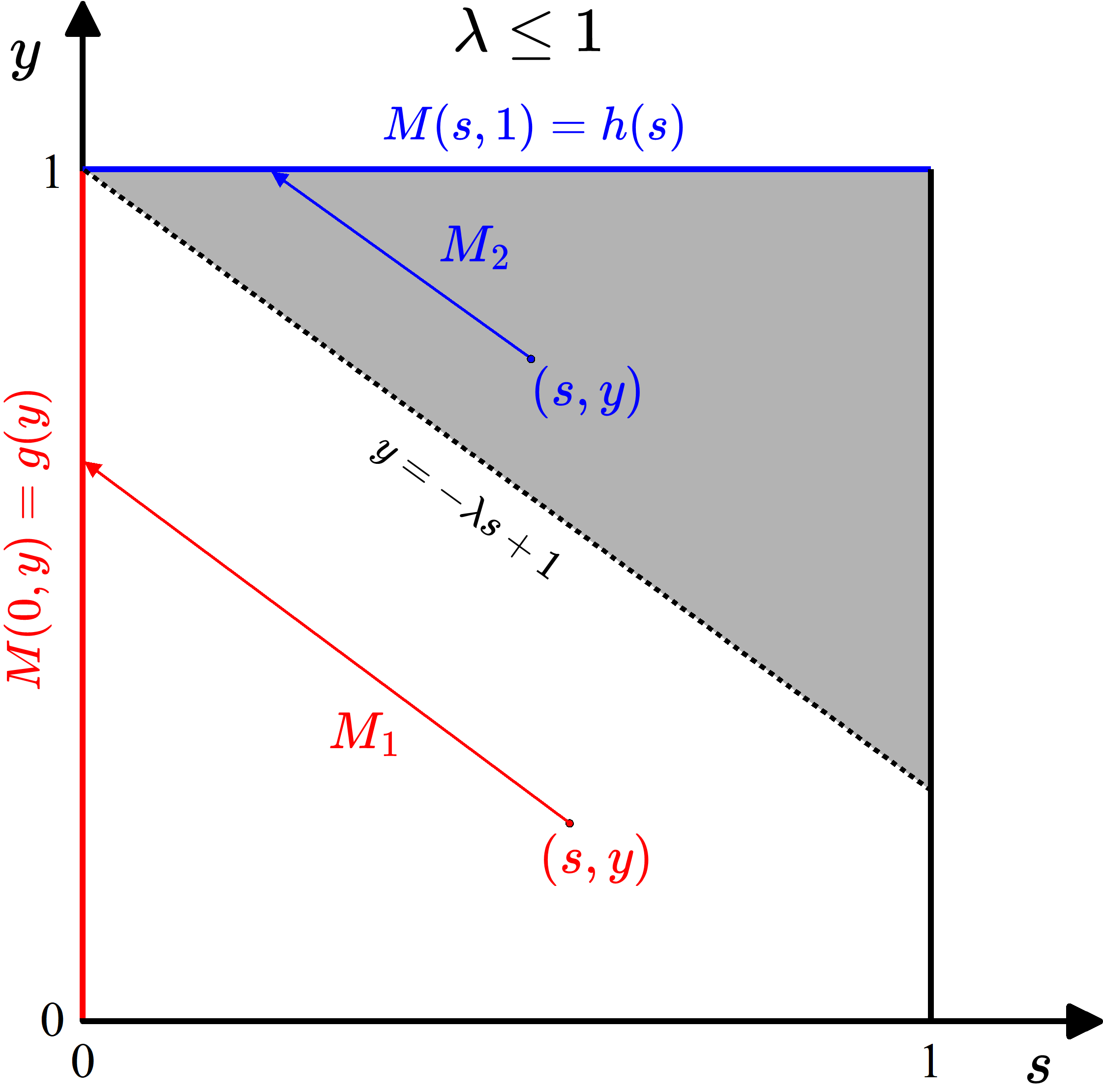}
                        \caption{}
                        \label{M smaller lambda}
                    \end{subfigure}
                    \hfill 
                    \begin{subfigure}[b]{0.47\linewidth}
                        \centering
                        \includegraphics[width=\linewidth,height=200pt,keepaspectratio=true]{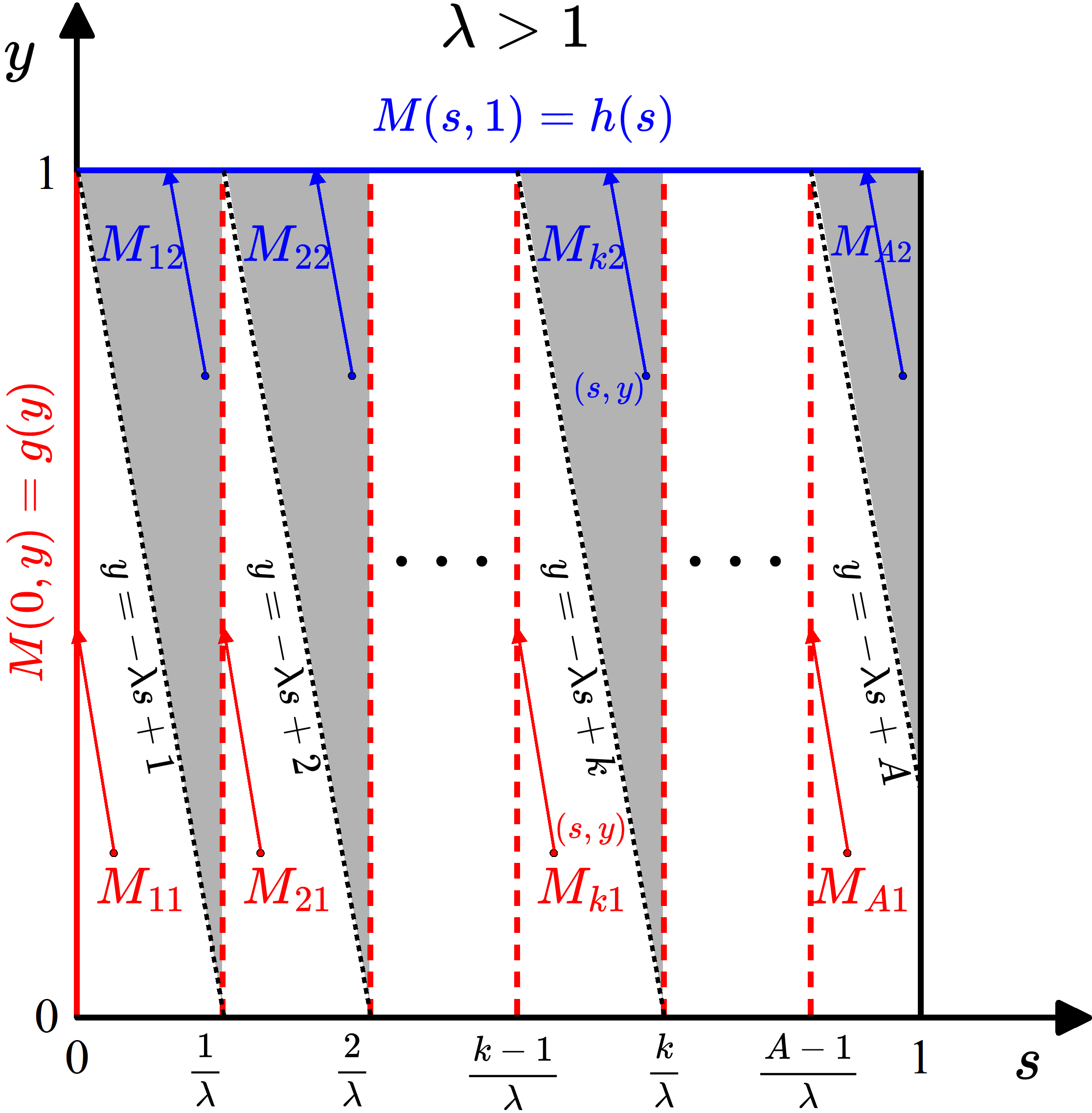}
                        \caption{}
                        \label{M larger lambda}
                    \end{subfigure}
                \end{center}
                \vspace{-8pt}
                \caption{Schematic of the domain of $M$, showing subregion division and characteristic lines.
                }
                \label{fig:kernel M lambda>1}
             \end{figure*}

\textbf{Case (i)} $\lambda\le 1$

Using the method of characteristics, we derive a two-branch integral form of $M$, arising from its two boundary conditions.
 \begin{align}
     M(s,y)=\left\{\begin{aligned}
         &M_1(s,y),& &\mathrm{if}~y\le -\lambda s+1, \\
         &M_2(s,y),& &\mathrm{if}~y> -\lambda s+1,
     \end{aligned}
     \right.
 \end{align}
 where 
 \begin{align}
     M_1(s,y) \!=& g(y+\lambda s)\!+\!\Phi_1[M_1](s,y), \label{M1 integral}\\
     M_2(s,y) \!=& h(\mu(s,y))\!+\!\Phi_2[M_2](s,y)+\!\Phi_3[M_1](s,y)\!, \label{M2 integral}
 \end{align}
 with $\mu(s,y)=s-\frac{1-y}{\lambda}$, and 
 \begin{align}
     \Phi_1[M_1]=&\int_{0}^{s}\!\!\!\int_{0}^{y+\lambda(s-\xi)}\!\!\!\!M_{1}(\xi,\theta)N(\theta,y+\lambda(s-\xi))d\theta d\xi, \\                  
     \Phi_2[M_2] =& \int_{0}^{\frac{1-y}{\lambda}}\!\!\!\int_{1-y-\lambda s}^{0}M_2(\xi+\mu(s,y),\theta-\lambda \xi+1) \notag \\
     &\qquad \quad \times N(\theta-\lambda \xi+1,-\lambda\xi+1)d\theta d\xi,\\
     \Phi_3[M_1] =& \int_{0}^{\frac{1-y}{\lambda}}\!\!\!\int_{0}^{2-y-\lambda(s+\xi)}M_1(\xi+\mu(s,y),\theta) \notag \\
     &\qquad \quad \times N(\theta,-\lambda\xi+1)d\theta d\xi.
 \end{align}
 
 We apply the method of successive approximations to obtain
 \begin{align}
     M_1^{n+1} =& g(y+\lambda s)+\Phi_1[M_1^{n}], \\
     M_2^{n+1} =& h(\mu(s,y))+\Phi_3[M_1]+\Phi_2[M_2^{n}],
 \end{align}
 with initial values $M_i^0(s,y)=0$, $i=1,\; 2$.
 Denote $\Delta M_i^{n}=M_i^{n+1}-M_i^{n}$, for $n=0,\; \ldots,\; \infty$ , then $\Delta M_1^0(s,y)=g(y+\lambda s)$, $\Delta M_2^0(s,y)=h(\mu(s,y))+\Phi_3[M_1](s,y)$, and 
 \begin{align}
     \Delta M_i^{n+1}(s,y) =& \Phi_i[\Delta M_i^{n}](s,y),\quad i=1,\;2. \label{Delta Mi n+1 and n}
 \end{align}
 First, we consider $M_1$. Assume that 
 \begin{align}
     &\left|\Delta M_1^{n}(s,y)\right|\le \bar g\frac{\bar N^ns^n}{n!}, \quad n\ge 0, \label{Delta M1 assume} 
 \end{align}
 where $\bar g=\|g\|_\infty$, $\bar N=\sup_{\mathcal{D}_1} |N(\theta,y)|$. Substituting \eqref{Delta M1 assume} into \eqref{Delta Mi n+1 and n} yields 
 \begin{align}
     \left|\Delta M_1^{n+1}(s,y)\right|\le& \int_{0}^{s}\!\!\!\int_{0}^{y+\lambda(s-\xi)}\bar g\frac{\bar N^n\xi^n}{n!}\bar Nd\theta d\xi \notag \\
     \le& \bar g\frac{\bar N^{n+1}}{n!}\int_{0}^{s}\xi^nd\xi 
     =\bar g\frac{\bar N^{n+1}s^{n+1}}{(n+1)!},
 \end{align}
 where we use $y+\lambda(s-\xi)\le  1$.
  Following the same procedure, we can obtain the following from $M_2$,
 \begin{equation}
     \left|\Delta M_2^{n}(s,y)\right|\le (\bar h\!+\!\bar \Phi_3)\frac{\bar N^n(y+\lambda s-1)^n}{n!}, \; n\ge 0, \label{Delta M2 assume}
 \end{equation}
 where $\bar h=\|h\|_{\infty}$, $\bar \Phi_3=\sup_{y\le -\lambda s+1} |\Phi_3[M_1](s,y)|$. 
 
 Therefore, the series 
 \begin{align}
     M_i(s,y)=&\sum_{n=0}^{\infty}\Delta M_i^n(s,y), \quad i=1,\;2, \label{Misum}
  \end{align}
 are uniformly convergent, which give 
  \begin{align}
     \left|M_1(s,y)\right|  \!\le \! \bar g\mathrm{e}^{\bar Ns},\: \left|M_2(s,y)\right|  \!\le\!\!  (\bar h \!+\! \bar \Phi_3)\mathrm{e}^{\bar N(y+\lambda s-1)}.
 \end{align}

\textbf{Case (ii)} $\lambda > 1$, corresponding to the input delay larger than one.

Since the characteristic line $y=-\lambda s+1$ intersects the boundary $y=0$, we can not prove the integral form of $M$ in the right region like the first case. In particular, in the area $s>\frac{1}{\lambda}$, shown in Fig.~\ref{M larger lambda}, the integral form derived from \eqref{M}-\eqref{M(0,y)} 
 \begin{align}
     M(s,y) \!=&h(\mu(s,y))\!+\!\int_{0}^{\frac{1-y}{\lambda}}\!\!\int_{0}^{-\lambda \xi+1}M(\xi+\mu(s,y),\theta) \notag\\
     &\times N(\theta,-\lambda \xi+1)d\theta d\xi\label{M integral equation not solved}
 \end{align}
 can not be formulated by a convergent series. 
 
 To address this issue, we partition the domain into several subregions and establish the existence of solutions successively within each. Denote an integer $A\ge2$ such that $A-1\le \lambda<A$ and define the following $2A$ triangular subregions in $\mathcal{D}_2$,
 \begin{align}
     &\mathcal{T}_k \!= \!\!\left\{\!(s,y)\!\in\!\left[\frac{k\!-\!1}{\lambda},\frac{k}{\lambda}\right]\!\times\![0,1]\right\}, \quad  k=1,\ldots, A, \label{Tk}\\
     &\mathcal{T}_{k1} \!\!= \!\!\left\{(s,y)\!\in\!\! \mathcal{T}_k : \frac{k-1}{\lambda} \le s\le\frac{k-y}{\lambda}\right\},\\
     &\mathcal{T}_{k2} \!\!= \!\!\left\{(s,y)\!\in\!\! \mathcal{T}_k : \frac{k-y}{\lambda} < s\le \min\left\{\frac{k}{\lambda},1\right\}\right\}.
 \end{align}

Denote $\varphi_k(s) =-\lambda s+k$ for $k= 1, \ldots, A$, then we define
 
 \begin{align}
     \Omega_{k1}[H_1] =&\int_{0}^{s-\frac{k-1}{\lambda}}\!\!\!\int_{0}^{y - \lambda\xi-\varphi_{k-1}(s)}\!\!H_1\Bigl(\xi\!+\!\frac{k\!-\!1}{\lambda},\theta\Bigr) \notag\\
     &\quad\:\times N(\theta,y\!-\!\lambda\xi-\varphi_{k-1}(s))d\theta d\xi, \label{Omega k1}\\                
     \Omega_{k2}[H_2]  =& \int_{0}^{\frac{1-y}{\lambda}}\!\!\!\int_{\varphi_k(s)-y}^{0}\!\!H_2(\xi\!+\!\mu(s,y),\theta \!-\! \lambda \xi+1) \notag \\
     &\quad \times N(\theta\!-\!\lambda \xi+1,-\lambda\xi+1)d\theta d\xi, \label{Omega k2}\\
     \Omega_{k3}[H_1]  =& \int_{0}^{\frac{1-y}{\lambda}}\!\!\!\int_{0}^{\varphi_{k+1}(s)-y-\lambda\xi}H_1(\xi+\mu(s,y),\theta) \notag \\
     &\quad \times N(\theta,-\lambda\xi+1)d\theta d\xi, \label{Omega k3}
 \end{align}
 where functions $H_1(s,y)$ and $H_2(s,y)$ are defined on $\mathcal{T}_{k1}$ and  $\mathcal{T}_{k2}$, respectively. 
 
 We denote function $M$ on the region $\mathcal{T}_{ki}$ by $M_{ki}(s,y)$.   
 From the boundary conditions \eqref{M(s,1)} and \eqref{M(0,y)}, we obtain the restriction of $M$  to $\mathcal{T}_{1}$,
\begin{align}\label{T1 solution}
     M(s,y)=\left\{\begin{aligned}
         &M_{11}(s,y),& &(s,y)\in\mathcal{T}_{11}, \\
         &M_{12}(s,y),& &(s,y)\in\mathcal{T}_{12},
     \end{aligned}
     \right.
\end{align}
 where 
 \begin{align*}
     M_{11}(s,y) =& g(y+\lambda s)\!+\!\Omega_{11}[M_{11}](s,y). 
     \\
     M_{12}(s,y) =& h(\mu(s,y))\!+\!\Omega_{13}[M_{11}](s,y)\!+\!\Omega_{12}[M_{12}](s,y), 
 \end{align*}
 The above equations become a replica of  \eqref{M1 integral}-\eqref{M2 integral}, which gives $M_{1i}\in C(\mathcal{T}_{1i})$, $i=1,\;2$.

 Assume $M_{k-1,i} \in C(\mathcal{T}_{k-1,i})$, $k=2,\ldots, A$,~$i=1,\; 2$, we will prove $M_{ki} \in C(\mathcal{T}_{ki})$, $i=1,\; 2$.
 Consider the next subregion $\mathcal{T}_k$, where $M_{k-1,2}(\frac{k-1}{\lambda}, y)$ that is already shown to be continuous and bounded serves as the boundary value for determining $M_{k1}$, as illustrated in Fig.~\ref{M larger lambda}. Using this boundary condition along with equation \eqref{M}, we can establish that $M_{k1} \in C(\mathcal{T}_{k1})$. Additionally, the function $M_{k2}$ depends solely on the boundary condition specified in \eqref{M(s,1)}, and it can likewise be shown to belong to $C(\mathcal{T}_{k2})$. 
 
 Specifically, we get the restriction of $M$ to $\mathcal{T}_k$ as 
 \begin{align}\label{Tk solution}
     M(s,y)=\left\{\begin{aligned}
         &M_{k1}(s,y),& &(s,y)\in\mathcal{T}_{k1}, \\
         &M_{k2}(s,y),& &(s,y)\in\mathcal{T}_{k2},
     \end{aligned}
     \right.
 \end{align}
 where 
 \begin{align}
     M_{k1} =& M_{k-1,2}\left(\frac{k-1}{\lambda}, y\!-\!\varphi_{k-1}(s)\right)\!+\!\Omega_{k1}[M_{k1}], \label{Mk1 integral}\\
     M_{k2} =& h(\mu)\!+\!\Omega_{k3}[M_{k1}]\!+\!\Omega_{k2}[M_{k2}], \label{Mk2 integral}
 \end{align}
 with $\Omega_{kj}$ defined in \eqref{Omega k1}-\eqref{Omega k3}.   

 We will show that each $M_{ki}$, $i=1,2$, can be expressed as a convergent series by using the successive approximation method. Define the following recursive formula with initial values $M_{ki}^0=0$, $i=1,\; 2$, 
 \begin{align}
     M_{k1}^{n+1} =& M_{k-1,2}\bigl(\frac{k-1}{\lambda}, y\!-\!\varphi_{k-1}(s)\bigr)\!+\!\Omega_{k1}[M_{k1}^{n}],\\
     M_{k2}^{n+1} =& h(\mu(s,y))\!+\!\Omega_{k3}[M_{k1}]\!+\!\Omega_{k2}[M_{k2}^{n}]. 
 \end{align}             
 Denote $\Delta M_{ki}^{n}(s,y)= M_{ki}^{n+1}(s,y)-M_{ki}^{n}(s,y)$, for $n=0,\; 1,\; \ldots,\; \infty$, which gives
 \begin{align}
     \Delta M_{ki}^{n+1}(s,y) =& \Omega_{ki}[\Delta M_{ki}^{n}](s,y),\quad i=1,\; 2. \label{Delta Mki n+1 and n}
 \end{align}
 By induction, we obtain, 
\begin{align}
     &\left|\Delta M_{k1}^{n}(s,y)\right|\le\bar M_{k-1,2}\frac{\bar N^n(s-\frac{k-1}{\lambda})^n}{n!}, \label{Delta Mk1 assume}
 \end{align}
 where $\bar M_{k-1,2}=\|M_{k-1,2}(\frac{k-1}{\lambda},y)\|_\infty$ and 
 \begin{align}
     &\left|\Delta M_{k2}^{n}(s,y)\right|\le(\bar h\!+\!\bar \Omega_{k3})\frac{\bar N^n(y-\varphi_k(s))^n}{n!},\label{Delta Mk2 assume}
 \end{align}
 where $\bar \Omega_{k3}=\sup_{\mathcal{T}_{k1}} |\Omega_{k3}[M_{k1}](s,y)|$. Consequently, the series
 \begin{align}
     M_{ki}(s,y) = \sum\limits_{n=0}^{\infty}\Delta M_{ki}^n(s,y), \quad i=1,\; 2\label{Mk1 Mk2 sum}
 \end{align}
 uniformly converge to the solution governed by \eqref{Mk1 integral}-\eqref{Mk2 integral}, with a bound
 \begin{align}
     &\left|M_{k1}(s,y)\right|  \le \bar M_{k-1,2}\mathrm{e}^{\bar N(s-\frac{k-1}{\lambda})}, \\
     &\left|M_{k2}(s,y)\right| \le (\bar h + \bar \Omega_{k3})\mathrm{e}^{\bar N(y-\varphi_k(s))}.
 \end{align}
Regarding the continuity, we follow the same procedure in \cite{coron2013local}. Since \eqref{Mk1 Mk2 sum} converges uniformly, we only need to prove continuity of each term. The initializations $\Delta M_{ki}^0$, $i=1,~2$, are continuous, then it can be shown immediately that $\Delta M_{ki}^n$ are continuous from \eqref{Delta Mki n+1 and n}.

Finally, combining \eqref{T1 solution} and \eqref{Tk solution}, we conclude that $M\in C(\mathcal{D}_2)$. This completes the proof of Lemma~\ref{well-posedness M}.             \hfill $\blacksquare$
        \end{pf}

 \begin{rem}\label{boundary compatibility}
        The equations \eqref{M}-\eqref{M(0,y)} remain well-posed without compatible boundary conditions, allowing for a piecewise continuous solution with a finite number of discontinuities along the characteristic line $y=-\lambda s+1$. Such incompatible boundary conditions have been discussed in \cite[Theorem A. 1]{hu2019boundary} and \cite[Remark 5]{hu2016control}. 
    \end{rem}
    \begin{rem}
   It should be noted that the domain-partitioning strategy used in Lemma~\ref{well-posedness M} differs from that in \cite{redaud2022stabilizing}. Here, the kernel domain is divided along characteristic lines into a finite number of triangular subdomains, and the solution is constructed successively on these subdomains by the method of successive approximations. In contrast, \cite{redaud2022stabilizing} partitions a rectangular kernel domain into rectangular or parallelogram-type subregions. Based on this domain-partitioning strategy, the proposed method can directly deal with kernel equations containing integral terms, regardless of whether these terms appear in the interior of the domain, on the boundary, or both, and enables the construction of explicit solutions for the kernel equations as a byproduct. Thus, unlike the approach in \cite{redaud2022stabilizing}, it does not require a preliminary transformation that transfers in-domain integral couplings to boundary integral couplings. 
    \end{rem}

    Now we present the well-posedness result of the kernel equations \eqref{J}-\eqref{G(0,y)}. 
    
    \begin{prop}\label{well-posedness J G L}
        Consider the kernel equations \eqref{J}-\eqref{G(0,y)}. 
        Under Assumption~\ref{assumption}, there exists a unique solution  in the $L^2$ sense, i.e., 
        $J,\, G \in L^2(\mathcal{D}_2)$. 
        Moreover, it can be further shown that $J,\, G$ is continuous in the case where $D_1 \le D_2$, i.e.,
        $J,\, G \in C(\mathcal{D}_2)$.
    \end{prop}    
    \begin{pf}
          The proof is divided into two cases: (i) $D_1 \le D_2$  and (ii) $D_1 > D_2$.  Before proceed, we define subregion $\mathcal{A}_1=\{(s,y)\in\mathcal{D}_2: y\ge \frac{D_1}{D_2}s\}$ and $\mathcal{A}_2=\{(s,y)\in\mathcal{D}_2: y < \frac{D_1}{D_2}s\}$.

        \textbf{Case (i)} $D_1 \le D_2$, corresponding to the case where input delay is less than state delay.   
        
        We first solve $G(s,y)$ in the subregion $\mathcal{A}_1$. Using the method of characteristics, combining the boundary condition \eqref{G(0,y)}, one gets
        \begin{align}
            G(s,y) =& F\left(0,y-\frac{D_1}{D_2}s\right),\quad (s,y)\in\mathcal{A}_1, \label{G A1 part}
        \end{align}
      which gives $G(s,1) = F(0,1-\frac{D_1}{D_2}s)$ and it is continuous.

 Substituting the resulting $G(s,1)$ into \eqref{J(s,1)}, we can decouple $J$ from $G$. Then applying Lemma~\ref{well-posedness M} by setting $\lambda=D_1$, $N(\theta,y)=D_1f(\theta,y)$, $h(s)=\frac{1}{D_2} F(0,1-\frac{D_1}{D_2}s)$, $g(y)=K(0,y)$, it gives  $J\in\mathcal{C}(\mathcal{D}_2)$. Notice that the boundary \eqref{J(s,1)} and \eqref{J(0,y)} are compatible at $(0,1)$, since 
        \begin{align}
            J(0,1)=\frac{1}{D_2}F(0,1)=K(0,1). 
        \end{align}

        Returning to $G(s,y)$ in the subregion $\mathcal{A}_2$, which is governed by the other boundary condition \eqref{G(s,0)}, we obtain
        \begin{align}
            G(s,y)\!=& D_2\int_0^1J\!\left(\!s\!-\!\frac{D_2}{D_1}y,\theta\!\right)\!c(\theta)d\theta, \, (s,y)\in\mathcal{A}_2, \label{G A2 part}
        \end{align}
        which gives $G\in C(\mathcal{A}_2)$. It is easy to validate the compatibility of the two boundary conditions of $G$ under Assumption~\ref{assumption}, that is 
        \begin{align}
            G(s,0)|_{s=0}=G(0,y)|_{y=0}=\int_0^1D_2K(0,y)c(y)dy.
        \end{align}

        \textbf{Case (ii)} $D_1 > D_2$.
        
        Note that in this case, the characteristic line $y=\frac{D_1}{D_2}s$ intersects the boundary $y=1$. The boundary $s=1$ extends across two regions, shown in the right panel of Fig.~\ref{fig:JG characteristic line case 2}. It implies that $G(s,1)$ is affected by both two boundary conditions  \eqref{G(s,0)} and \eqref{G(0,y)}. Consequently, $J(s,1)=G(s,1)/D_2$ consists of two branches separated by $s=\frac{D_2}{D_1}$. In the second subregion, where $s>\frac{D_2}{D_1}$, we have
        \begin{align}
            J(s,1) = \int_0^1 J\left(s-\frac{D_2}{D_1},\theta\right)c(\theta)dy,\;  \mathrm{for}~s>\frac{D_2}{D_1}, \label{non-trivial}
        \end{align}
        which is obtained by substituting equation \eqref{G A2 part} into \eqref{J(s,1)}.
        This nonlocal boundary condition brings difficulties to the well-posedness analysis, since the presence of the shifting term $J(s-\frac{D_2}{D_1},\theta)$ in the integral prevents the construction of a convergent series representation for $J$.
        \begin{figure}[htpb]
            \centering
            \includegraphics[width=1.05\linewidth]{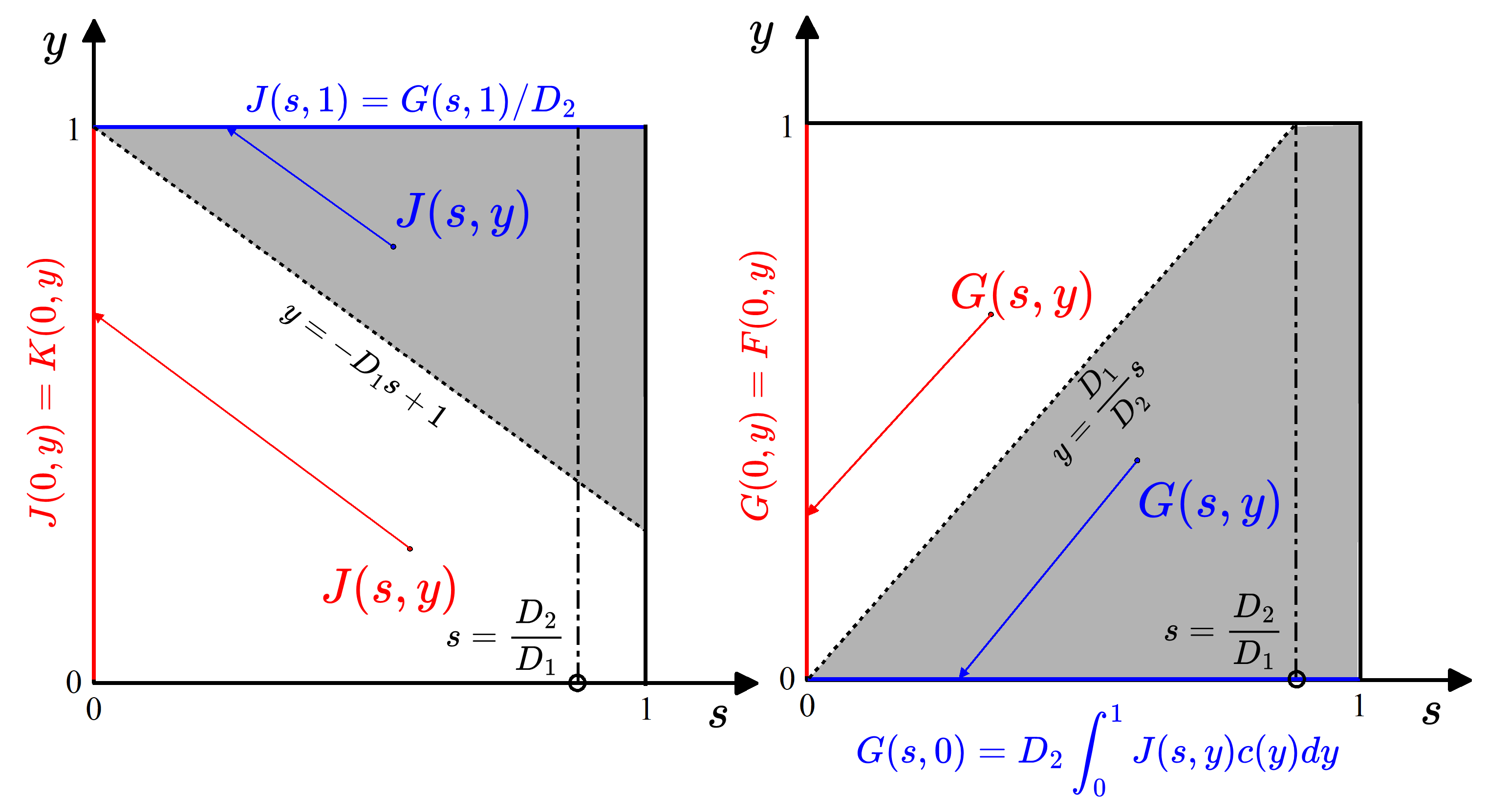}            
            \caption{The domain of $J$ (Left) and $G$ (Right) and the corresponding characteristic lines as $D_1>D_2$.  
            }
            \label{fig:JG characteristic line case 2}
        \end{figure}
        
   To address this issue, we adopt the infinite induction energy series method originally proposed in \cite{chen2024backstepping}. This method constructs an infinite hierarchy of energy functionals associated with recursively defined auxiliary functions. By proving the convergence of this energy series, it proves the well-posedness of the complicated coupled PDEs with nonlocal boundary conditions.

        Specifically, we construct infinite sequences of equations whose solutions are denoted by $\{j^{n}(s,y)\}$, $\{g^{n}(s,y)\}$ for $n=1, \ldots, \infty$,
        \begin{subequations}\label{jn gn}
        \begin{align}
            j_{s}^{n}(s,y) =& D_{1}j_{y}^{n}(s,y)+D_{1}\int_{0}^{y}j^{n}(s,\theta)f(\theta,y)d\theta, \label{jn}\\
            j^{n}(s,1) =& \frac{1}{D_{2}}g^{n}(s,1), \quad j^n(0,y)=0, \label{jn boundary initial}\\
            g_{s}^{n}(s,y) =& -\frac{D_{1}}{D_{2}}g_{y}^{n}(s,y), \label{gn}\\
            g^{n}(s,0) =&D_2\int_0^1j^{n-1}(s,y)c(y)dy, ~  g^{n}(0,y)=0.\label{gn boundary initial}
        \end{align}
        \end{subequations}
        with initial equation only coupled in a single direction,
        \begin{subequations}\label{j0 g0}
        \begin{align}
            j_{s}^{0}(s,y) =& D_{1}j_{y}^{0}(s,y)+D_{1}\int_{0}^{y}j^{0}(s,\theta)f(\theta,y)d\theta, \label{j0}\\
            j^{0}(s,1) =& \frac{1}{D_{2}}g^{0}(s,1), \quad j^0(0,y)=K(0,y), \label{j0 boundary initial}\\
            g_{s}^{0}(s,y) =& -\frac{D_{1}}{D_{2}}g_{y}^{0}(s,y), \label{g0}\\
            g^{0}(s,0) =&0, \qquad\qquad\quad  g^{0}(0,y)=F(0,y).\label{g0 boundary initial}
        \end{align}
        \end{subequations}
        It is clear that the initial equation is well-posed from Lemma~\ref{well-posedness M}. If \eqref{jn gn} and \eqref{j0 g0} have a solution, then $J(s,y)$ and $G(s,y)$ which satisfy \eqref{J}-\eqref{G(0,y)} can be represented as
        \begin{align}
            J=&\sum_{n=0}^{\infty}j^n(s,y), & G=&\sum_{n=0}^{\infty}g^n(s,y), \label{J G sum}
        \end{align}
        where each $j^n(s,y)$ and $g^n(s,y)$ is solution of the equations \eqref{j0 g0} and \eqref{jn gn}.
        Moreover, the explicit solution of $g_0$ and $g_n$ can be written as 
        \begin{align}\label{g0 solution}
                 g^0(s,y)=\left\{\begin{aligned}
                     &F\Bigl(0,y-\frac{D_1}{D_2}s\Bigr),& & ~\text{in}~\mathcal{A}_1, \\
                     &0,& &~\text{in}~\mathcal{A}_2.
                 \end{aligned}
                 \right.
        \end{align}        
        and 
        \begin{align}\label{eq:gn}
            g^n(s,y) \!=\! \left\{\!\begin{aligned}
                     &0, \qquad \qquad \qquad \qquad \qquad \qquad  ~\text{in}~\mathcal{A}_1, \\
                     &\!\!\int_0^1\!\! D_2j^{n-1}\Bigl(s\!-\!\frac{D_2}{D_1}y,\theta\Bigr)c(\theta)d\theta, ~\text{in}~\mathcal{A}_2, 
                 \end{aligned}
                 \right.
        \end{align}  
        which indicates that \eqref{jn}-\eqref{jn boundary initial} define well-posed $j^n$ from Lemma~\ref{well-posedness M}.

        Introduce the following sequence of functionals $V_n(s)$ with respect to $j^n$, $g^n$ for $n\ge 1$,
        \begin{align}
            V_{1}(s) =& \int_{0}^{1}\left(j^{1}(s,y)\right)^{2}dy+\gamma\int_{0}^{1}\left(g^{1}(s,y)\right)^{2}dy, \label{sequence V1}\\
            V_{n}(s) =& V_{n-1}(s)+\int_{0}^{1}\left(j^{n}(s,y)\right)^{2}dy \notag \\
            &+\gamma\int_{0}^{1}\left(g^{n}(s,y)\right)^{2}dy, \qquad \mathrm{for}~n\ge2, \label{sequence Vn}
        \end{align}
        where $\gamma>0$ is a constant to be determined later. 
                
        Differentiating \eqref{sequence V1} with respect to $s$ along the solution of \eqref{jn gn} yields
        \begin{align}
            V_1^{\prime}(s) \! =&\! \int_0^1\!\! 2D_{1}j^1(s,y)\left[j_{y}^{1}(s,y)\!+\!\!\int_{0}^{y}\!\! j^{1}(s,\theta)f(\theta,y)d\theta\right]dy \notag \\
            &+\gamma\int_0^1 2g^1(s,y)\left(-\frac{D_{1}}{D_{2}}g_{y}^{1}(s,y)\right)dy \notag \\
            =& -\frac{D_1}{D_2}\left(\gamma-\frac{1}{D_2}\right)\left(g^1(s,1)\right)^2-D_1\left(j^1(s,0)\right)^2 \notag \\
            &+\int_0^1\!\!\!\int_{0}^{y}2D_{1}j^1(s,y)j^{1}(s,\theta)f(\theta,y)d\theta dy \notag \\
            &+\gamma D_1D_2\left(\int_0^1j^{0}(s,y)c(y)dy\right)^2,            
        \end{align}
       where the last line follows from \eqref{gn boundary initial}. Choosing $\gamma\ge\frac{1}{D_2}$ and using the Cauchy-Schwarz inequality, one gets
        \begin{align}
            V_1^{\prime}(s) \le& \alpha_1 \!\int_0^1\!\!\! \left(j^{1}(s,y)\right)^{2}dy \!+\! \alpha_0 \!\int_0^1\!\!\! \left(j^{0}(s,y)\right)^{2}dy \notag \\
            \le& \alpha_1 \!\int_0^1\!\!\! \left(j^{1}(s,y)\right)^{2}dy \!+\! \beta 
            \le \alpha_1V_1(s)+\beta, \label{V1 estimate}
        \end{align}
        where $\alpha_1=2D_1\bar f$, $\alpha_0=\gamma D_1D_2\bar c^2$, $\beta=\alpha_0\bar{j^0}^2$, $\bar c=\|c\|_\infty$, $\bar f=\sup_{\mathcal{D}_1}|f(x,y)|$, $\bar{j^0}=\sup_{\mathcal{D}_2}\lvert j^0(s,y)\rvert$. 
        
        In a similar method, we derive that for $n=2,\ldots$
        \begin{align}\label{eq:V_n}
          V_{n}^{\prime}(s) \le& V_{n-1}^{\prime}(s) \! +\!  \int_0^1 \alpha_1(j^{n}(s,y))^{2}+ \alpha_0(j^{n-1}(s,y))^{2} dy.
        \end{align}
 By summing these inequalities from $1$ to $n$, we further derive, 
        \begin{align}
            V_{n}^{\prime}(s) \le& \int_0^1\alpha_1(j^{n}(s,y))^{2}dy + \int_0^1\alpha\sum_{k=1}^{n-1}(j^k(s,y))^2 + \beta,\notag \\\le&\int_0^1\alpha\sum_{k=1}^{n}(j^k(s,y))^2 + \beta, \label{Vn general formula}
        \end{align}
where $\alpha=\alpha_1+\alpha_0$. Recalling 
      \begin{align}\label{Vn summation form}
            V_n(s)\!=\!\int_0^1\!\sum_{k=1}^{n}(j^k(s,y))^2\!+\!\gamma (g^k(s,y))^2dy, ~n\ge1,
        \end{align}
we finally obtain
 \begin{align}
      V_{n}^{\prime}(s)   \le \alpha V_n(s)+\beta.
 \end{align}
Therefore, 
\begin{equation}
    V_n(s)\le\frac{\beta}{\alpha}(\mathrm{e}^{\alpha s}-1),
\end{equation}
implying that the sequence $\{V_n(s)\}$ has an upper bound and is thus convergent, then by the dominated convergence theorem \cite{chen2024backstepping}, the series in \eqref{J G sum} converge and admit a solution in  $ L^2(\mathcal{D}_2)$. This completes the proof of Proposition~\ref{well-posedness J G L}.
        \hfill $\blacksquare$
    \end{pf}

Note that the infinite induction energy series method used in the second case is also applicable to the first case. Nevertheless, the method of characteristics provides a continuous solution so we use it to address the first case.

    \begin{rem}\label{statement on assumption}
        Without Assumption~\ref{assumption}, the two boundary conditions become incompatible. As discussed in \cite{hu2016control} and \cite{hu2019boundary}, such incompatibility along the characteristic line still permits the existence of piecewise continuous kernel solutions. Therefore, Assumption~\ref{assumption} can be relaxed, allowing piecewise continuous solutions for the kernel equations \eqref{K}-\eqref{F(x,0)} and \eqref{J}-\eqref{G(s,0)}.
    \end{rem}

\section{Stability analysis}\label{stability analysis}
    In this section, we prove the main result Theorem~\ref{main result}.    
    The proof relies on the stability of the target system and the norm equivalence between the original and the target systems. 
   
    \subsection{Stability of target system}
    \begin{lem}\label{target stability}
        Consider the target system \eqref{target system} with initial conditions $w(x,0)=w_0$, $z(s,0)=z_0$, $p_0\in L^2(0,1)$ being compatible with the boundary conditions. Then the target system is exponentially stable in the $L^2$ sense at the origin, satisfying
        \begin{align}
            V_1(t)\le C_1\mathrm{e}^{-\sigma_1 t}V_1(0),
        \end{align}
        where $C_1$ and $\sigma_1$ are positive constants, and 
        \begin{align}
        V_1(t)=\|w(\cdot,t)\|^2+\|z(\cdot,t)\|^2+\|p(\cdot,t)\|^2. \label{V2}
        \end{align}
        Additionally the origin is reached in finite time $t=t_\mathrm{F}$, given by \eqref{tF}.  
    \end{lem} 
    \begin{pf}
        The solution of \eqref{target system} is given by
        \begin{align*}
            z\!=&z_0(s\!+\!\frac{t}{D_1}),~t<\! \phi_1(s);~z=0,~ t\ge \!\phi_1(s), \notag \\
            w\!=&w_0(x-t),~ t< x;~w=z(0,t-x),t\ge x, \notag \\
            p\!=&p_0(s\!+\!\frac{t}{D_2}),~t\!<\! \phi_2(s),~p\!=\!w(1,t\!-\!\lambda_2(s)),~t\!\ge \!\phi_2(s),
        \end{align*}
        where $\phi_i(s)=D_i(1-s)$, $i=1,\; 2$. It is clear that $z\equiv0$ for $t\ge D_1$. Consequently, $w$ vanishes for $t\ge D_1+1$, and $p$ reaches the origin for $t\ge D_2+D_1+1$. Therefore, the state reaches the origin within a finite time $t_\mathrm{F}$, from which exponential stability follows directly. This completes the proof. \hfill $\blacksquare$
    \end{pf}

    \subsection{Inverse transformation and norm equivalence}
       We postulate the inverse transformation, 
      \begin{align}\label{inverse trans}
        \begin{bmatrix}
            p \\ u \\ v
        \end{bmatrix}\!\!=\!&
        \begin{bmatrix}
            I & \mathbf{0} & \mathbf{0} \\
            \breve{\mathscr{F}}_1 & \breve{\mathscr{L}}_1+I & \mathbf{0} \\
            \breve{\mathscr{F}}_2 & \breve{\mathscr{F}}_3 & \breve{\mathscr{L}}_2+I
        \end{bmatrix}\!\!
        \begin{bmatrix}
            p \\ w \\ z
        \end{bmatrix}
    \end{align}      
    where $\breve{\mathscr{F}}_1: p\mapsto \int_{0}^{1}\breve F(x,y)p(y)dy$, $\breve{\mathscr{F}}_2: p\mapsto \int_{0}^{1}\breve G(s,y)\\p(y)dy$, $\breve{\mathscr{F}}_3: w \mapsto\int_{0}^{1}\breve J(s,y)w(y)dy$, $\breve{\mathscr{L}}_1: w\mapsto \int_{x}^{1} \!\breve K(x,y)w(y)dy$, and $\breve{\mathscr{L}}_2: z\mapsto \int_{0}^{s} D_1\breve J(s-y,0)z(y)dy$, $\breve K$, $\breve F$, $\breve J$, and $\breve G$ are defined on the same domains as $K$, $F$, $J$, and $G$, respectively.  
    
    Inserting \eqref{inverse trans} into \eqref{u trans} and \eqref{v trans}, we obtain
    \begin{align}
        \breve K(x,y)=&K(x,y)+\int_x^yK(x,\theta)\breve K(\theta,y)d\theta, \label{breve K}\\
        \breve F(x,y)=&F(x,y)+\int_x^1K(x,\theta)\breve F(\theta,y)d\theta,\label{breve F}
    \end{align}
    and
     \begin{align}
        \breve J(s,y)=&J(s,y)+\int_0^yJ(s,\theta)\breve K(\theta,y)d\theta \notag \\
        &+\int_0^sD_1J(s-\theta,0)\breve J(\theta,y)d\theta,\\
        \breve G(s,y)=&G(s,y)+\int_0^1J(s,\theta)\breve F(\theta,y)d\theta \notag \\
        &+\int_0^sD_1J(s-\theta,0)\breve G(\theta,y)d\theta.
    \end{align}
    It is clear the above equations of the inverse kernel functions are the second-kind Volterra equations \cite{linz1985analytical}, implying that the inverse kernels exist and  are continuous, given the direct kernel $K$, $F$, $J$ and $G$. 

    \begin{lem}\label{norm equivalence}
        Consider the transformation \eqref{u trans}-\eqref{v trans} with its inverse \eqref{inverse trans}. There exist $\kappa_i>0$, $i=1,\;2,\;3,\;4$ such that
        \begin{subequations}
        \begin{align}
            \|w(\cdot,t)\|^2\le&\kappa_1\left(\|u(\cdot,t)\|^2+\|p(\cdot,t)\|^2 \right),\label{w estimate}\\
            \|z(\cdot,t)\|^2\le&\kappa_2\left(\|v(\cdot,t)\|^2\!+\!\|u(\cdot,t)\|^2\!+\!\|p(\cdot,t)\|^2 \right),   \label{z estimate}  \\       
            \|u(\cdot,t)\|^2\le&\kappa_3\left(\|w(\cdot,t)\|^2+\|p(\cdot,t)\|^2 \right),\label{u estimate}\\
            \|v(\cdot,t)\|^2\le&\kappa_4\left(\|z(\cdot,t)\|^2\!+\!\|w(\cdot,t)\|^2\!\!+\!\|p(\cdot,t)\|^2 \right).\label{v estimate}
        \end{align}
        \end{subequations}
    \end{lem}
    \begin{pf}
        From \eqref{u trans}, using the Cauchy-Schwarz inequality and omitting time-dependence, one gets
        \begin{align}
            \|w(\cdot,t)\|^2
            \!\le&3\|u(\cdot,t)\|^2\!+\!3\int_0^1\left[\int_{x}^{1} K(x,y)u(y)dy\right]^2dx \notag \\
            &+3\int_0^1\left[\int_{0}^{1} F(x,y)p(y)dy\right]^2dx\notag \\
            \!\le&3\|u(\cdot,t)\|^2\!+\!\int_0^1\!\!\!\int_{x}^{1} \!3K^2(x,y)dy\!\!\int_{x}^{1}\!\! u^2(y)dydx \notag \\
            &+\int_0^1\!\!\!\int_{0}^{1} \!3F^2(x,y)dy\!\!\int_{0}^{1}\!\! p^2(y)dydx \notag \\
            \!\le&3(1\!+\!\bar K^2)\|u(\cdot,t)\|^2\!+\!3\bar F^2\|p(\cdot,t)\|^2,\!
        \end{align}
        where $\bar K^2\!\!=\!\!\int_0^1\!\!\int_{x}^{1} \!K^2(x,\!y)dydx$, $\bar F^2\!\!=\!\!\int_0^1\!\!\int_{0}^{1} \!F^2(x,\!y)dydx$. The remaining inequalities are proved similarly using \eqref{v trans} and \eqref{inverse trans}. \hfill $\blacksquare$
    \end{pf}
    
    Combining Lemma~\ref{target stability} and Lemma~\ref{norm equivalence} concludes the proof of Theorem~\ref{main result}.

\section{Numerical simulation}\label{numerical simulation}
    To validate the effectiveness of the proposed controller \eqref{controller}, we present a numerical example. In the example, the plant coefficients of \eqref{extended system} are set as $c(x)=\sin \pi x$, $f(x,y)=\mathrm{e}^{x+y^2}$, and two cases of delays are considered, $(D_1,~D_2)=(1,~2)$ for case (i), $(D_1,~D_2)=(2,~1)$ for case (ii).

    The finite difference method is employed to numerically solve the PDE system \eqref{extended system}, where the temporal and spatial step sizes are sets as: $\Delta t=0.001$, $\Delta x=0.01$. The initial conditions are $u_0(x)=\sin 2\pi x$, $v_0(s)=0.02\sin \pi s$, $p_0(s)=\cos 3\pi s/2$, satisfying the compatible conditions. The controller \eqref{controller} is numerically implemented applying the Simpson method using the gain kernels $J(s,y)$ and $G(s,y)$ which are numerically solved by the finite difference method with discretized step size $\Delta s=0.01$.

    The system is open-loop unstable, as shown in Fig. \ref{subfig:u open}. Fig. \ref{fig:u close} illustrates the closed-loop simulation results. The results demonstrate that the closed-loop system is successfully stabilized under the delay-compensated controller \eqref{controller}, whereas it diverges when applying the nominal controller (see \cite{krstic2008backstepping}), as evidenced in Fig.~\ref{subfig:u uncomp}. Observe that the open-loop system diverges more slowly. This occurs because, in first-order hyperbolic PDEs, instantaneous recirculation has a more destabilizing effect than delayed recirculation due to their transport properties.
    \begin{figure}[htpb]
        \begin{center}
            \begin{subfigure}[b]{0.49\linewidth}
                \includegraphics[width=\linewidth]{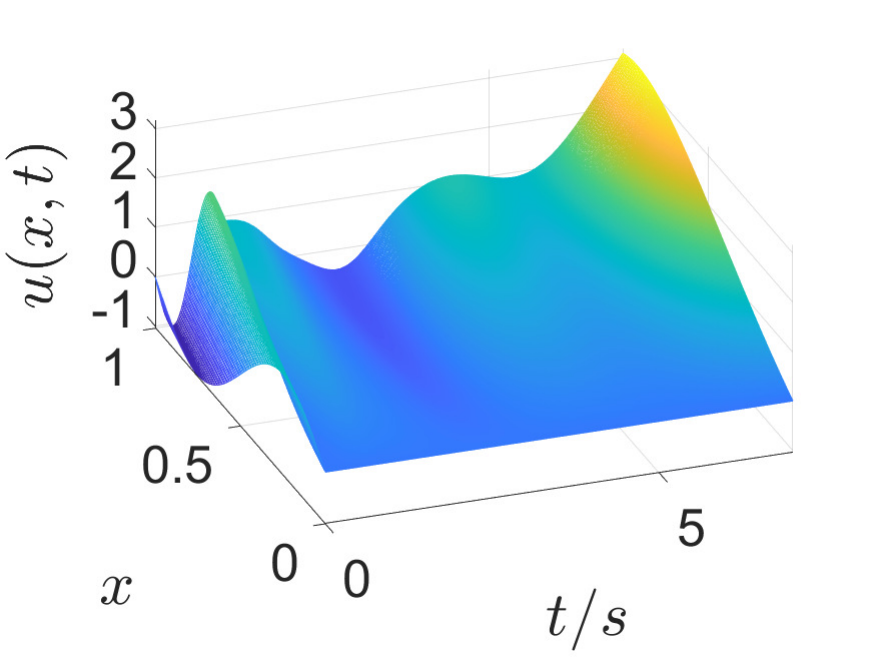}
                \caption{}
                \label{subfig:u open}
            \end{subfigure}
            \hfill 
            \begin{subfigure}[b]{0.49\linewidth}
                \includegraphics[width=\linewidth]{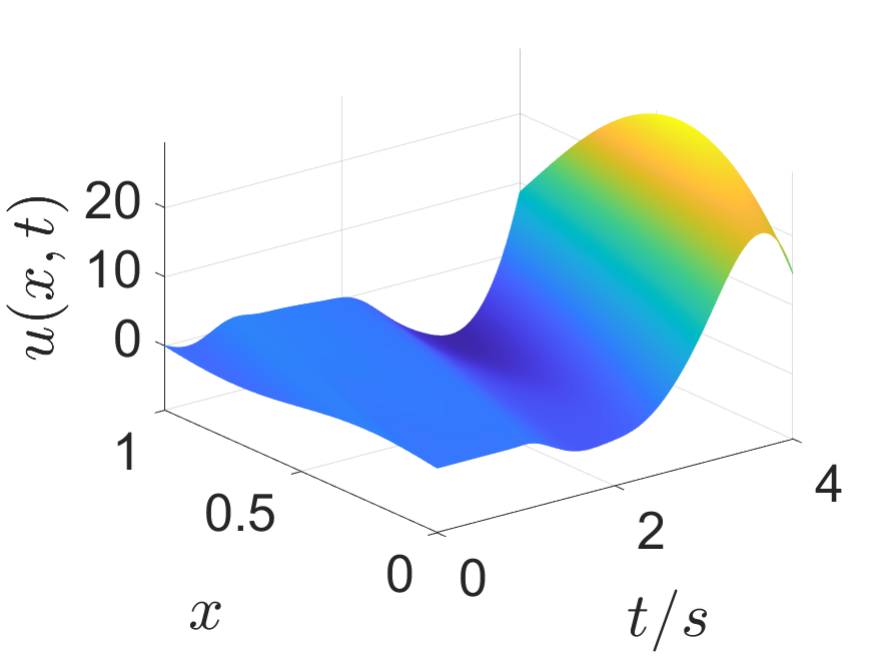}
                \caption{}
                \label{subfig:u uncomp}
            \end{subfigure}
        \end{center}
        \caption{The dynamics of the state $u(x,t)$. \subref{subfig:u open}: in open-loop configuration. \subref{subfig:u uncomp}: delay uncompensated.}
        \label{fig:u open, u uncomp}
    \end{figure}
    \begin{figure}[htpb]
        \begin{center} 
            \begin{subfigure}[b]{0.49\linewidth}
                \includegraphics[width=\linewidth]{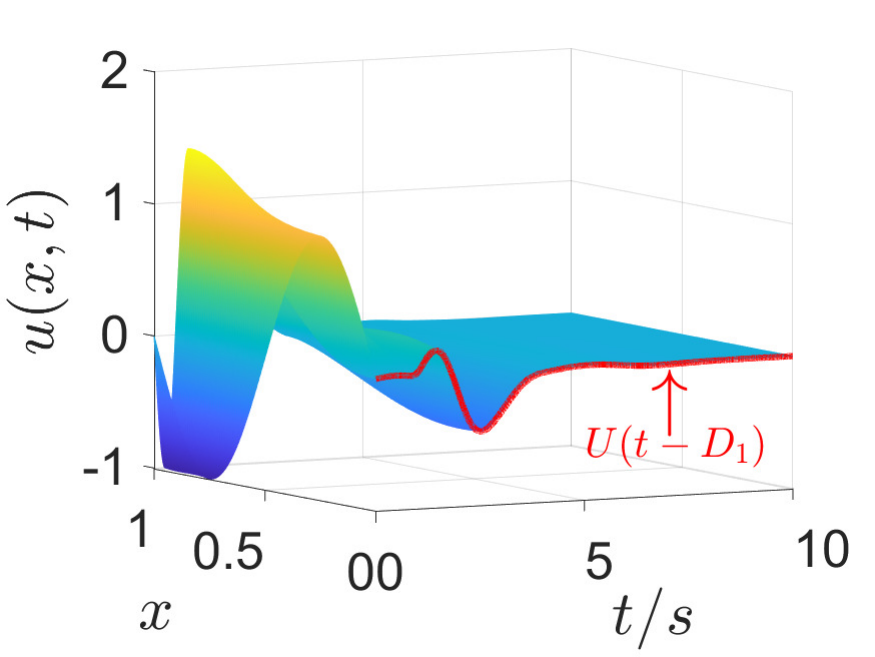}
                \caption{}
                \label{subfig:u close 1}
            \end{subfigure}
            \hfill
            \begin{subfigure}[b]{0.49\linewidth}
                \includegraphics[width=\linewidth]{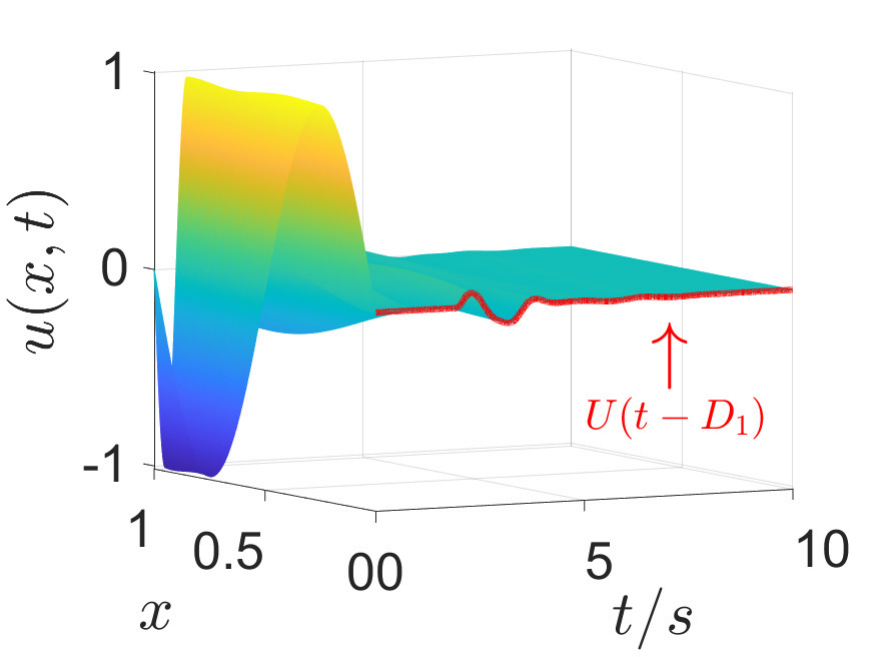}
                \caption{}
                \label{subfig:u close 2}
            \end{subfigure}                        
        \end{center}
        \caption{The dynamics of the state $u(x,t)$ with delay compensation in two delay cases. \subref{subfig:u close 1}: $D_1=1$, $D_2=2$. \subref{subfig:u close 2}: $D_1=2$, $D_2=1$.}
        \label{fig:u close}
    \end{figure}

    For a clearer illustration, we plot the $L^2$ norm of the state $u(x,t)$ under the delay-compensated controller in two cases in Fig. \ref{subfig:u norm}, and the time evolution of the boundary controller in Fig. \ref{subfig:controller}.
     \begin{figure}
        \begin{center}
            \begin{subfigure}[b]{\linewidth}
                \includegraphics[width=\linewidth]{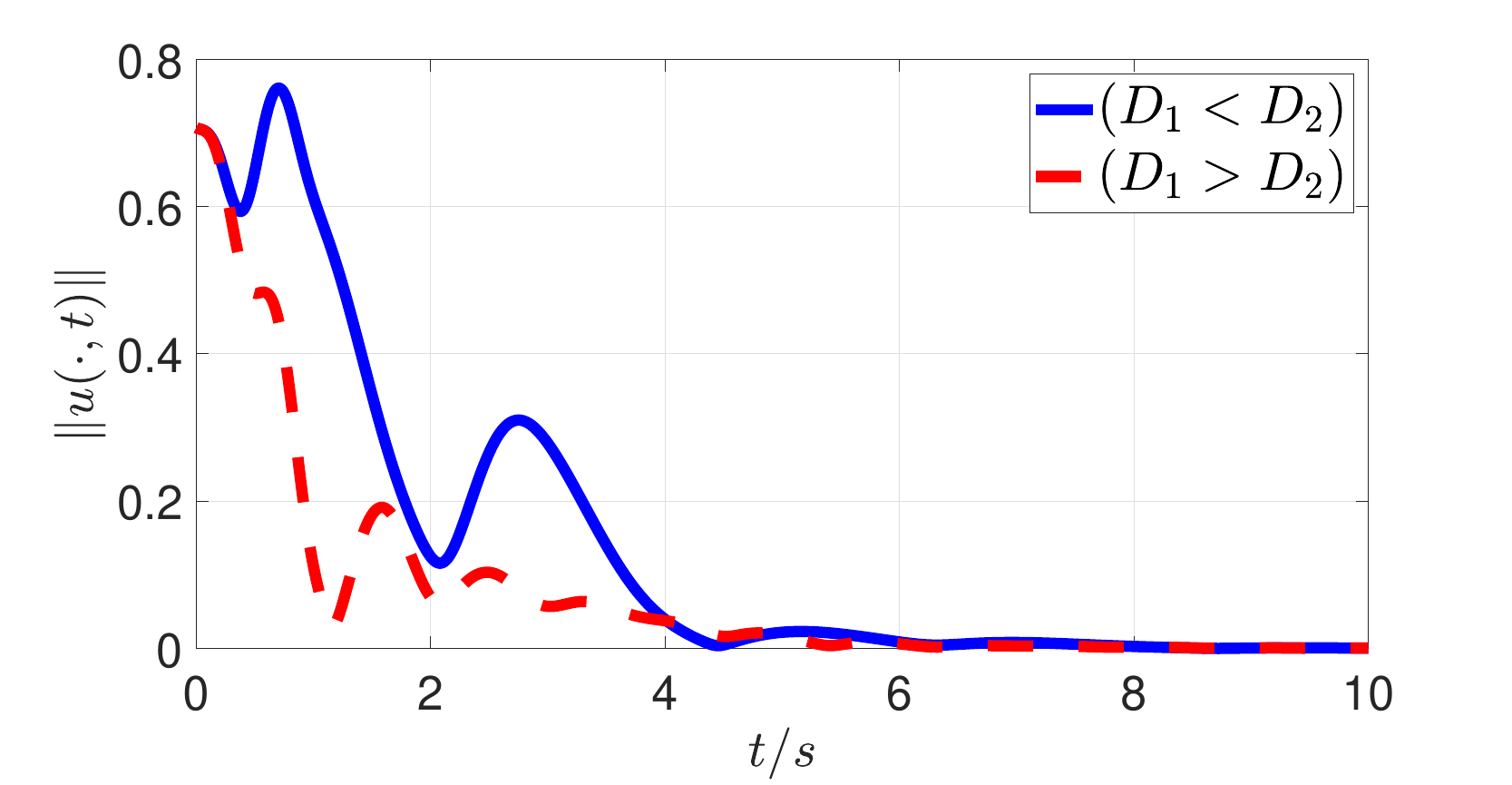}
                \caption{}
                \label{subfig:u norm}
            \end{subfigure}

            \begin{subfigure}[b]{\linewidth}
                \includegraphics[width=\linewidth]{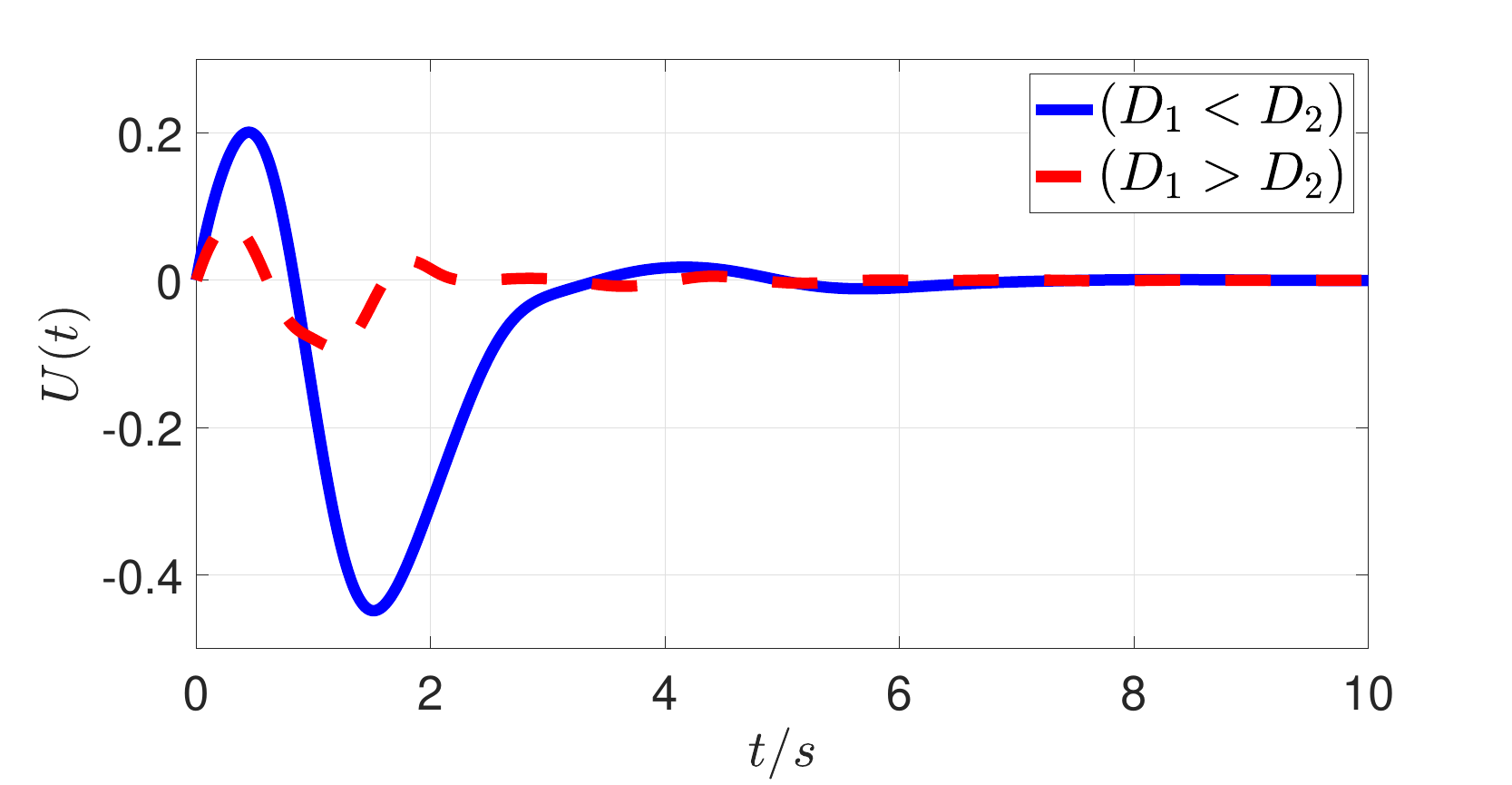}
                \caption{}
                \label{subfig:controller}
            \end{subfigure}
        \end{center}
        \caption{The $L^2$ norm of the state $u(x,t)$ and evolution of the controller with delay compensation in two delay cases. \subref{subfig:u norm}: state norm $\|u(\cdot,t)\|$, \subref{subfig:controller}: controller $U(t)$.}
        \label{fig:u norm, controller}
     \end{figure}
    
    Fig. \ref{fig:mismatch result} presents the simulation result for the system with delay mismatch in case (i) ($(D_1,~D_2)=(1,~2)$). For underestimated-delay cases, i.e., the actual delay is greater than the assumed delay, $\Delta D_1=\Delta D_2=0.1$ are chosen, and for overestimated-delay cases, $\Delta D_1=\Delta D_2=-0.1$. The results show that the controller is robust to certain delay mismatches. 
    \begin{figure}
        \begin{center}
            \begin{subfigure}[b]{\linewidth}
                \includegraphics[width=\linewidth]{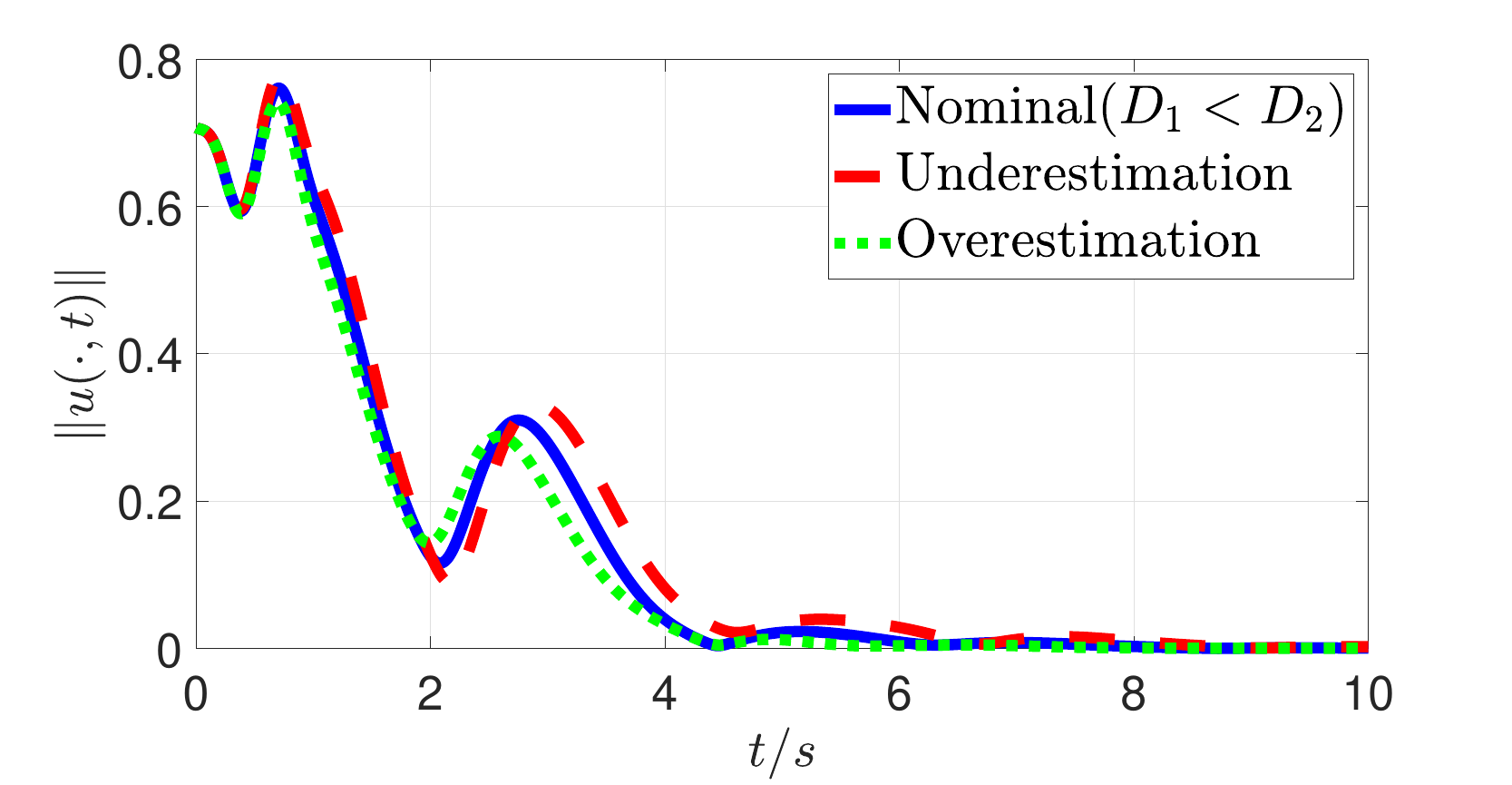}
                \caption{}
                \label{subfig:u norm mismatch}
            \end{subfigure}

            \begin{subfigure}[b]{\linewidth}
                \includegraphics[width=\linewidth]{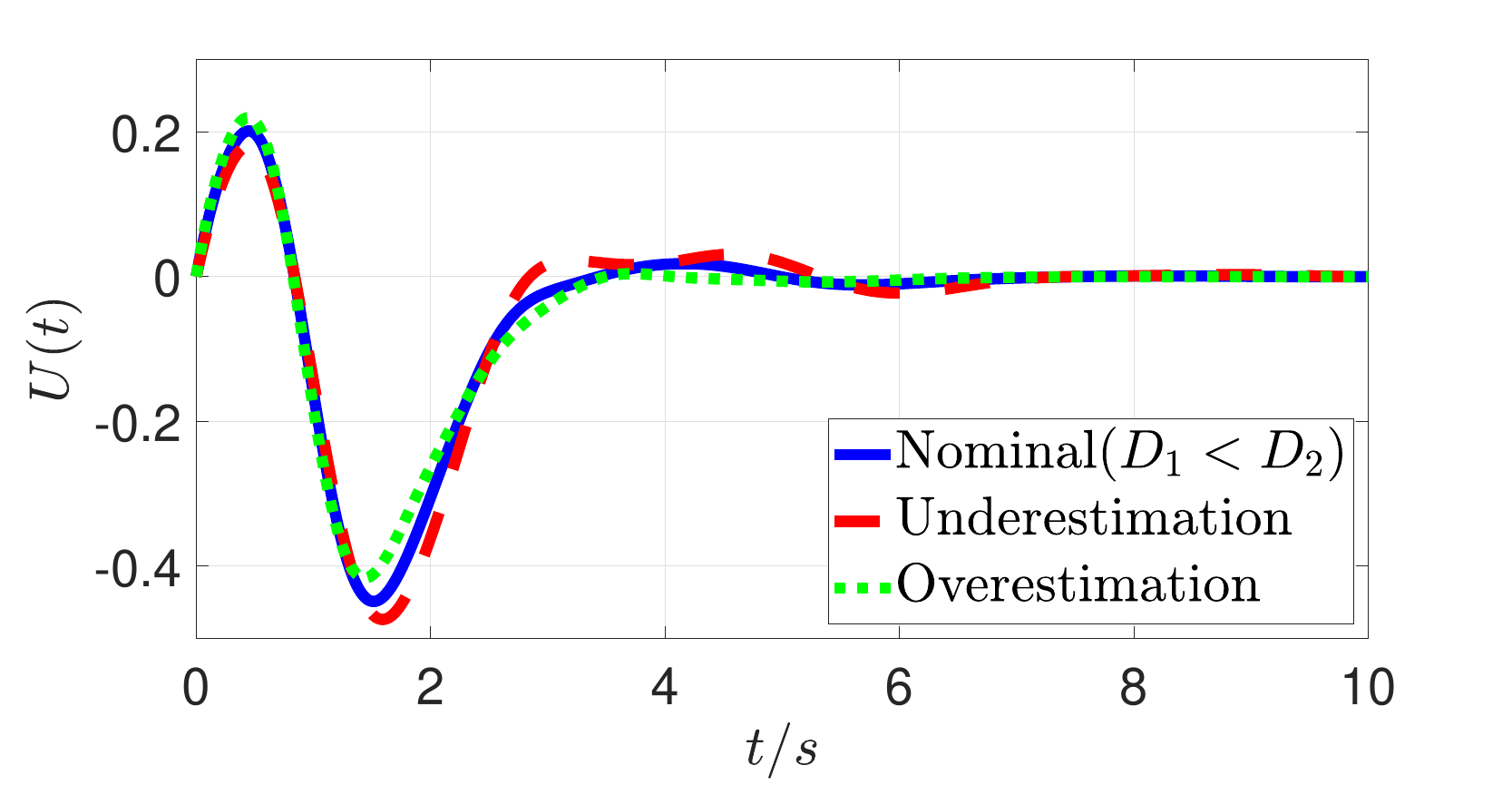}
                \caption{}
                \label{subfig:controller mismatch}
            \end{subfigure}
        \end{center}
        \caption{The $L^2$ norm of the states $u(x,t)$ and evolution of the controller with delay mismatch under nominal delay values $D_1=1$, $D_2=2$. Dashed line: underestimation. Dotted line: overestimation. \subref{subfig:u norm mismatch}: state norm $\|u(\cdot,t)\|$, \subref{subfig:controller mismatch}: controller $U(t)$.}
        \label{fig:mismatch result}
    \end{figure}

\section{Conclusion}\label{conclusion}
    A backstepping-based control strategy is developed to stabilize systems of linear hyperbolic PIDEs subject to simultaneous and arbitrarily large input and state delays. We design two invertible integral transformations, each involving both Fredholm- and Volterra-type integrals. For the resulting kernel equations in the form of a system of four coupled PIDEs, we divide the well-posedness analysis into two cases depending on whether input delay is smaller or larger than state delay, which are solved by employing the successive approximation method and the infinite induction energy series method, respectively. The finite-time stability of the closed-loop system is established. Finally, a numerical example is provided. Future work will consider uncertain delay and the more general case where state delays appear in the source term.

\bibliographystyle{elsarticle-num}
\bibliography{refs}

\end{document}